\documentclass[draftcls,onecolumn,12pt]{IEEEtran}

\usepackage{amsmath,amstext,amsthm,amssymb,epsf}
\usepackage{latexsym}
\usepackage[pdftex]{graphicx}
\usepackage{setspace}
\numberwithin{equation}{section} 

 \newtheorem{lemma}{Lemma}[section]
 \newtheorem{theorem}[lemma]{Theorem}

 \newtheorem{definition}[lemma]{Definition}
 \newtheorem{rem}[lemma]{Remark}
\newenvironment{remark}{\begin{rem}}{\hspace*{\fill}$\diamondsuit$\end{rem}}
 \newtheorem{ex}[lemma]{Example}

\begin{document}

\title{Normalized Google Distance of Multisets with Applications}
\author{Andrew R. Cohen\thanks
{Andrew Cohen is with the Department of Electrical and Computer Engineering, 
Drexel University.
Address: A.R. Cohen, 3120--40 Market Street, 
Suite 313, Philadelphia, PA 19104,
USA. Email: {\tt acohen@coe.drexel.edu}} 
and Paul M.B. Vit\'{a}nyi
\thanks{
Paul Vit\'{a}nyi is with the national research center for mathematics 
and computer science in the Netherlands (CWI),
and the University of Amsterdam.
Address:
CWI, Science Park 123,
1098XG Amsterdam, The Netherlands.
Email: {\tt Paul.Vitanyi@cwi.nl}.
}}


\maketitle

\begin{abstract}
Normalized Google distance (NGD) is a relative semantic distance
based on the World Wide Web 
(or any other large electronic database, 
for instance Wikipedia) 
and a search engine that returns aggregate page counts. 
The earlier NGD between
pairs of search terms (including phrases) is not sufficient 
for all applications.
We propose an NGD of finite multisets of 
search terms that is better for many applications. 
This gives a relative semantics shared by a multiset of search
terms. We give applications and compare the results
with those obtained using the pairwise NGD. The derivation 
of NGD method
is based on Kolmogorov complexity.

\end{abstract}

\section{Introduction}
\label{sect.intro}

The classical notion of 
Kolmogorov complexity \cite{Ko65} is an objective measure 
for the information in an 
a {\em single} object, and information distance measures the information 
between a {\em pair} of objects \cite{BGLVZ98}. 
Information distance leads to
the {\em normalized compression distance} (NCD) \cite{Li03,CV04}
by normalizing it to obtain
the so-called similarity metric, and subsequently 
approximating the Kolmogorov complexity through real-world
compressors.
This NCD is
a parameter-free, feature-free, and alignment-free
similarity measure  that
has found many applications in pattern recognition, phylogeny, clustering,
and classification, see the many references in Google scholar to \cite{CV04}.

There arises the question of the
shared information between many objects
instead of just a pair of objects.  
To satisfy such aims the information distance 
measure has been extended from pairs to nonempty finite multisets \cite{Li08}. 
More properties of this distance were investigated in \cite{Vi11}.
For certain applications we require a normalized version. 
For instance,
classifying an object into one or another of several classes we aim
for the class of which the NCD for multisets grows the least.
The new NCD was applied to classification questions that were
earlier treated with the pairwise NCD. The results obtained 
were significantly better \cite{CV13}. 

Up till now the objects considered can be viewed as 
finite computer files that carry all their properties within.
However, there are also objects that carry all their properties
without like `red' or `three' or are not a computer
files like `Einstein' or `chair.' Such objects are represented by name.
Some objects can be viewed as either text or name, 
such as the text of ``Macbeth by Shakespeare''
or the name ``Macbeth by Shakespeare.''

In the name case we define a similarity distance
based on the background information
provided by Google or any search engine that produces aggregate page 
counts. Such search engines discover the ``meaning'' of words and 
phrases relative to other words and phrases in the sense
of producing a relative semantics \cite{CV07}. There, the distance between
the search terms $x$ and $y$ is given as 
the {\em normalized Google distance} (NGD) 
by 
\begin{equation}\label{eq.ngd}
 NGD(x,y)
=  \frac{  \max \{\log f(x), \log f(y)\}  - \log f(x,y) }{
\log N - \min\{\log f(x), \log f(y) \}},
\end{equation}
where $f(x)$ denotes the number of pages containing occurrences
of $x$, $f(x,y)$
denotes the number of pages containing occurrences
of both $x$ and $y$, and $N$ denotes the
total number of web pages indexed by Google 
(or a multiple thereof, see below). We use the binary logarithm
denoted by ``$\log$'' throughout. The NGD is widely applied,
for example \cite{BMI07,GKAH07,XWF08,WM09,CL10} and the many references to
\cite{CV07} in Google scholar.

\subsection{An Example}
On 9 April 2013 googling for ``Shakespeare'' gave 130,000,000 hits;
googling for ``Macbeth'' gave 26,000,000 hits; and googling
for ``Shakespeare'' and ``Macbeth'' together gave 20,800,000 hits.
The number of pages indexed by Google was estimated by the number
of hits of the search term ``the'' which was 25,270,000,000 hits. Assuming
there are about 1,000 search terms on the average page this gives 
$N=25,270,000,000,000$.
Hence $NGD(Shakespeare, Macbeth) = (26.95 - 24.31)/(44.52-24.63)=0.13$.
According to the formula $NGD(x,y)=0$ means $x$ is the same as $y$ (rather the
implication $\max\{f(x),f(y)\}= f(x,y)$) and $NGD(x,y)=1$ means
$x$ is as unlike $y$ as is possible. Hence Shakespeare and Macbeth are
very much alike according to the relative semantics supplied by Google.

\subsection{Related Work}\label{sect.relwork}

In \cite{Li08} the notion is introduced of the information required to
go from any object in a finite multiset of objects to any other object in the
multiset. Let $X$ denote a finite multiset of 
$n$ finite binary strings defined by (abusing the set notation)
$\{x_1, \ldots, x_n\}$, the constituting
elements (not necessarily all different) 
ordered length-increasing lexicographic.
We use multisets and not sets, since in a set all 
elements are different while here we are interested in the situation were
some or all of the elements are equal. 
Let $U$ be the reference universal Turing machine, for convenience the prefix
one \cite{LV08}.
We define the {\em information distance} in $\{x_1, \ldots, x_n\}$
by $E_{\max} (x_1, \ldots , x_n) = \min \{|p|: U(x_i,p,j)=x_j$ for all 
$x_i,x_j, \; 1 \leq i,j \leq n \}$.
It is shown in \cite{Li08}, Theorem 2, that 
\begin{equation}\label{eq.li08}
E_{\max}(x_1, \ldots , x_n)=
\max_{x:x \in \{x_1, \ldots , x_n\}} K(x_1, \ldots , x_n|x),
\end{equation}
up to a logarithmic additive term.

The information distance in \cite{BGLVZ98}
between strings $x_1$ and $x_2$ is denoted $E_{\max}(x_1,x_2) 
= \max\{K(x_1|x_2),K(x_2|x_1)\}$.
In \cite{Vi11} we introduced the notation $X=\{x_1, \ldots, x_n\}$ so that
$E_{\max}(X) = \max_{x:x \in X} K(X|x)$.
The two notations coincide for $|X|=2$ since $K(x,y|x)=K(y|x)$ up
to an additive constant term. The quantity $E_{\max}(X)$ is called 
the information distance in $X$. It comes in two flavors:
the pairwise version for $|X|=2$ and the multiset version for $|X| \geq 2$. 
The normalized pairwise version was made computable using real-world 
compressors to approximate the incomputable Kolmogorov complexity.
Called the normalized compression distance (NCD) it has turned out to be
suitable for determining similarity between pairs of objects, for phylogeny,
hierarchical clustering,  heterogeneous data clustering, anomaly detection,
and so on \cite{Li03,CV04}. Applications abound.
In \cite{CV07} the name case for pairs was resolved by using
the World Wide Web as database and Google as query mechanism 
(or any other search engine that give an
aggregate page count). Viewing the search engine as a compressor
and using the NCD formula this gives many new applications. 

The theory of information distance for multisets insofar it was not
treated in \cite{Li08} was given in \cite{Vi11}. 
In \cite{CV13} the $E_{\max}$ distance of nonempty finite multisets was normalized 
and approximated by real-world compressors. The result is the {\em normalized
compression distance} (NCD) for multisets. The $NCD(X)$ where $X$ is
a multiset is shown to be a metric with values in between 0 and 1. 
The developed theory was applied to classification questions concerning
the fate of retinal progenitor cells, synthetic versions, organelle transport,
and handwritten character recognition (a problem in OCR). In all cases
the results were significantly better than using the pairwise NCD except
for the OCR problem where a combination of the two approaches gave 99.6\%
correct on MNIST data. The current state of the art classifier 
for MNIST data achieves 99.77\% accuracy. 

\section{Results}
We translate the NCD approach for multisets in \cite{CV13}
to the relative semantics case investigated for the relative semantics of
pairs of words (or phrases) in \cite{CV07}. This gives the relative 
semantics of a multiset of names (or phrases).
The basic concepts like the Google distribution
and the Google code are given in Section~\ref{sect.google}.
In Section~\ref{sect.theory} we give the relevant properties 
like non-metricity. We show that the closer the
google probability of a multiset approximates the universal
probability of that multiset the closer the NGD approximates
a normalized form of the information distance, with equality of the
latter for equality of the former.
(A normalized form of information distance 
quantizes the common similarity 
of the objects in the multiset according to all effective properties.)
We subsequently apply the NGD for multisets to various data sets,
colors versus animals, saltwater fish versus freshwater fish,
taxonomy of reptiles, mammals, birds, fish, and US Primary
candidates in Section~\ref{sect.appl}. Here we compare the 
outcomes with those of the pairwise NGD in \cite{CV07} using
different search engines Google, Bing, and 
the Google n-gram
data base. We show that the multiset NGD is as good or superior
on these examples, except sometimes for Bing (the 
Google n-gram method did not currently work for the multiset NGD
since Google supplies n-grams for n is at most 5 and the multisets
inspected have too large cardinality).

\subsection{Strings}
We write {\em string} to mean a finite binary string,
and $\epsilon$ denotes the empty string.
(If the string is over
a larger finite alphabet we recode it into binary.)
The {\em length} of a string $x$ (the number of bits in it)
is denoted by $|x|$. Thus,
$|\epsilon| = 0$.

\subsection{Sets, Multisets, and Kolmogorov Complexity}\label{sect.set} 
A {\em multiset} is a generalization of the notion of set.
The members are allowed to appear more than once.
For example, if $x\neq y$ then $\{x,y\}$ is a set, but $\{x,x,y\}$
and $\{x,x,x,y,y\}$ are multisets, with abuse of the set notation.
We also abuse the
set-membership notation by using it
for multisets by writing $x \in \{x,x,y\}$ and $z \not\in \{x,x,y\}$
for $z \neq x,y$. Further, $\{x,x,y\} \setminus \{x\} = \{x,y\}$. 
If $X,Y$ are multisets then we use the notation $XY=X \bigcup Y$;
if $Z$ is a nonempty multisets and $X=YZ$, then
we write $Y \subset X$.
For us, a multiset is finite and nonempty such as
$\{x_1 , \ldots ,x_n\}$ with $0 < n < \infty$ and the members are
finite binary strings in
length-increasing lexicographic order. If $X$ is a multiset, then some
or all of its elements may be equal. $x_i \in X$ means
that ``$x_i$ is an element
of multiset $X$.''
With $\{x_1 ,\ldots , x_n\} \setminus \{x\}$ we mean the multiset
$\{x_1 \ldots x_n\}$
with one occurrence of $x$ removed.

The finite binary strings, finiteness,
and length-increasing lexicographic order allows us to assign
a unique Kolmogorov complexity to a multiset.
The conditional prefix Kolmogorov complexity $K(X|x)$ of a multiset
$X$ given an element $x$ is the length of a shortest program $p$
for the reference universal Turing machine that with input $x$
outputs the multiset $X$. The prefix Kolmogorov complexity $K(X)$ of a multiset
$X$ is defined by $K(X| \epsilon )$.
One can also put multisets in the conditional such as $K(x|X)$ or
$K(X|Y)$.
We will use the straightforward laws $K(\cdot|X,x)=K(\cdot|X)$
and $K(X|x)=K(X'|x)$ up to an additive constant term, for $x \in X$
and $X'$ equals the multiset $X$ with one occurrence of the element $x$ deleted.

\section{Google Distribution and Google Code}
\label{sect.google}
We repeat some relevant text from \cite{CV07} since it is as true 
now as it was then.
\begin{quotation}
The number of web pages currently indexed by Google is approaching
$10^{11}$. Every common search term occurs in millions of web pages.
This number is so vast, and the number of web authors generating
web pages is so enormous (and can be assumed to be a
truly representative very large sample from humankind),
that the probabilities of Google search terms, conceived as
the frequencies of page counts returned by Google divided by
the number of pages indexed by Google (multiplied by the average number
of search terms in those pages),
approximate the actual relative frequencies of those search terms
as actually used in society. Based on this premise,
the theory we develop in this paper states
that the relations represented by the
normalized Google distance or NGD \eqref{eq.NGD}
approximately capture
the assumed true semantic
relations governing the search terms.
The NGD formula 
only uses the probabilities of search terms
extracted from the text corpus in question. We use the World Wide Web
and Google. The same method may be used with other text corpora
like Wikipedia, the King James version of the
Bible or the Oxford English Dictionary and frequency count
extractors, or the World Wide Web again and Yahoo as frequency
count extractor. In these cases one obtains a text corpus and
frequency extractor based semantics of the search terms.
To obtain the true relative frequencies of words and phrases
in society is a major problem in applied linguistic research.
This requires analyzing representative random samples of sufficient
sizes. The question of how to sample randomly and
representative is a continuous source of debate.
Our contention that the World Wide Web is such a large and diverse text corpus,
and Google such an able extractor, that the relative
page counts approximate the true societal word- and phrases
usage, starts to be supported by current real linguistics research.
\end{quotation}

\subsection{The Google Distribution:}
Let the set of singleton {\em Google search terms}
be denoted by ${\cal S}$ and $s=|{\cal S}|$.
If a set search term has $n$ singleton search terms then there are
${s \choose n}$ such set search terms.
\begin{remark}
\rm
There are
${ s  \choose n}$ ($1 \leq n \leq s$)  set search terms 
consisting of $n$ non-identical terms and hence \[
\sum_{1 \leq n \leq s} { s  \choose n} = 2^{s}-1
\]
set search terms altogether. 
\end{remark}
However, for practical reasons mentioned in the opening paragraph
of Section~\ref{sect.gsd} we use multisets instead of sets.
\begin{definition}
\rm
Let $X$ be a multiset of search terms defined by 
$X=\{x_1, \ldots, x_n\}$ with
$x_i \in {\cal S}$ for $1 \leq i \leq n$,
and ${\cal X}$ be the set of such $X$.
\end{definition}
Let the set of web pages indexed (possible of being returned)
by Google be $\Omega$. The cardinality of $\Omega$ is denoted
by $M=|\Omega|$, and at the time of this writing we estimate 
$25 \cdot 10^{9} \leq M \leq 100 \cdot 10^{9}$. (Google 
does not anymore report the number of web pages indexed. Searching for
common words like ``the'' or ``a'' one gets a lower bound on this
number.) 
A subset of $\Omega$
is called an {\em event}. Every search term $X$ usable by Google
defines a {\em Google event} ${\bf X} \subseteq \Omega$ of web pages
that contain an occurrence of $X$ and are returned by Google
if we do a search for $X$.
If $X=\{x_1, \ldots , x_n\}$, then ${\bf X}={\bf x}_1 \bigcap
\cdots \bigcap {\bf x}_n 
\subseteq \Omega$ is the set of web pages returned by Google
if we do a search for pages containing the search terms $x_1$ through
$x_n$ simultaneously. (But see the caveat in opening paragraph 
of Section~\ref{sect.gsd}.)
\begin{remark}
\rm
We can also define the other Boolean combinations: $\neg {\bf x}=
\Omega \backslash {\bf x}$ and ${\bf x}_1 \bigcup \cdots \bigcup {\bf x}_n =
\neg ( \neg {\bf x}_1 \bigcap \cdots \neg {\bf x}_n)$.
For example, let ${\bf e}$ is the event obtained from the basic events ${\bf x}_1, {\bf x}_2,
\ldots$, corresponding to basic search terms $x_1,x_2, \ldots$,
by finitely many applications of the Boolean operations.
\end{remark}

\subsection{Google Semantics}
Google events capture in a particular sense
all background knowledge about the search terms concerned available
(to Google) on the web.
\begin{quote}
The Google event ${\bf x}$, consisting of the set of
all web pages containing one or more occurrences of the search term $x$,
thus embodies, in every possible sense, all direct context
in which $x$ occurs on the web. This constitutes the Google semantics
of the term $x$.
\end{quote}
\begin{remark}
\rm
It is of course possible that parts of
this direct contextual material link to other web pages in which $x$ does not
occur and thereby supply additional context. In our approach this indirect
context is ignored. Nonetheless, indirect context may be important and
future refinements of the method may take it into account.
\end{remark}
\subsection{The Google Code}
The event ${\bf X}$ consists of all
possible direct knowledge on the web regarding $X$.
Therefore, it is natural
to consider code words for those events
as coding this background knowledge. However,
events of singleton search terms may overlap 
and hence the summed probabilities based on page counts divided by the
total number of pages exceeds 1.
The reason is that,
for example, the search terms ``cats'' and ``dogs'' often
will occur on the same page.
By the Kraft-McMillan inequality \cite{Kr49,Mc56} this prevents a
corresponding set of uniquely decodable code-word lengths. Namely
to be uniquely decodable the binary code words $c_1,c_2, \ldots$ must satisfy
\[
\sum_i 2^{-|c_i|} \leq 1.
\]
Note that this overlap phenomenon only concerns the singleton search terms.
Multiple search terms based on Boolean combinations of 
singleton search terms in a query return a page count based
on the same Boolean combination of singleton search terms events. 

The solution to the overlap problem is to count every
page as the number of search terms that occur in it. 
For example, page $p \in \Omega$ contains $p_s$ singleton search terms
one or more times. Then,
the sum of those counts is
\[
N = \sum_{p \in \Omega} p_s = \sum_{x \in S} |{\bf x}|.
\]
Since every web page that is indexed by Google contains at least
one occurrence of a search term, we have $N \geq M$. On the other hand,
web pages contain on average not more than a certain constant $\alpha$
search terms. Therefore, $N \leq \alpha M$.

\begin{definition}
\rm
For multisearch term $X$ corresponding with event ${\bf X}$
define the {\em Google distribution} by
\begin{align}\label{eq.gpmf}
g(X)  =|{\bf X}|/N. 
\end{align}
\end{definition}
\begin{remark}
\rm
Possibly search terms $X \neq Y$ but ${\bf X}={\bf Y}$.
In this case $|{\bf X}|=|{\bf Y}|=|{\bf X} \bigcap {\bf Y}|$
and $NGD(X,Y)=0$ according to the later formula \eqref{eq.NGD}. 
This case is treated in Lemma~\ref{lem.0}.
In this case the NGD formula views search multisets $X$ and $Y$
as being identical. This is not correct but a quirk of the formalism
that, luckily, does not occur frequent.
\end{remark}
With $\alpha$ large enough we have 
$\sum_{X \in {\cal X}} g(X) \leq 1$.
This $g$-distribution changes over time,
and between different samplings
from the distribution. But let us imagine that $g$ holds
in the sense of an instantaneous snapshot. The real situation
will be an approximation of this.
Given the Google machinery, these are absolute probabilities
which allow us to define the associated prefix code-word 
lengths (information contents) equal to unique decodable code words
length \cite{Mc56} for the multiset search terms.
\begin{definition}\label{def.g}
\rm
The length $G(x)$ of the {\em Google code} for a multiset search 
term $X$ is defined by (rounded upwards)
\begin{align}\label{eq.gcc}
G(X)= \log 1/g(X) ,
\end{align}
or $\infty$ for $g(X)=0$.
The case $|X|=1$ gives the length of the 
Google code for the singleton search terms.
This is the length of the Shannon-Fano code \cite{Sh48}.
\end{definition}

\subsection{The Google Similarity Distance}\label{sect.gsd}
In the information distance in \cite{BGLVZ98} 
one uses objects which are computer
files that carry their meaning within. It makes sense for $X$
to be a multiset of such objects. But
in the Google case it appears to make no sense to have 
multiset occurrences of
the same singleton search term in a query. Yet the frequency count
returned by Google for ``red'' was on July 24, 2013, 4.260.000.000,
for ``red red''  5.500.000.000, for ``red red red''  5.510.000.000.
The count keeps rising with more occurrences of the same term
until say 10 to 15 occurrences. Because of this perceived quirk
we use multisets $X$ instead of restricting $X$ to be a set. 

The information distance is expressed in terms of prefix Kolmogorov complexity.
One approximates the prefix Kolmogorov complexity from above by a real-world
compressor $Z$. For every multiset $X$ of computer files we have
$K(X) \leq Z(X)$ where $Z(X)$ is the length of the compressed version of $X$ 
using the compressor $Z$. See \cite{Li03,CV04,CV13}. 

We have seen that a multiset search
term  (a.k.a. a query) $X$ 
induces a probability $g(X)$ and similarly a
Google code of length 
$G(X)=\log 1/g(X)$. 
In this sense we see Google as a compressor, say $G$,
for the Google semantics associated with the search terms.
Replace the compressor $Z$ by Google viewed as compressor $G$.
\begin{remark}
\rm
Using $\max_{x \in X}\{G(X \setminus \{x\})$ in the denominator of
\eqref{eq.NGD} instead of $\max_{x \in X}\{G(x)\}$ is a valid choice.
Denote the result of using the denominator mentioned first as $NGD'(X)$.
In this case Lemma~\ref{lem.0} would state that the range of $NGD'(X)$
is from 0 to 1.  
As $|X|$ rises the value of $NGD'(X)$ tends to 1, similar to
the case of the multiset NCD in \cite{CV13}. 
That is, the numerator
tends to become equal to the denominator. In the applications
Section~\ref{sect.appl} we want to classify $x$ in one of the classes
$X_i$ ($i \in I$ where $I$ is small). To this purpose, 
we would look in the $NGD'(X)$ case for the $i$ that
minimizes  $NGD'(X_i\bigcup \{x\})-NGD'(X_i)$. This difference tends
to become smaller when $|X_i|$ rises. For $|X_i| \rightarrow \infty$
the difference goes to 0. Hence we choose as denominator 
$\max_{x \in X}\{G(x)\}$. 

The information distance is universal (\cite{BGLVZ98} for pairs and 
\cite{Vi11} for multisets) in the sense that
every admissible distance of $X$ (every such distance quantifies a 
shared property of the elements
of $X$ satisfying certain mild restrictions) is at least as large 
as that numerator. The information distance is an admissible distance. Therefore
the numerator is called {\em universal} \cite{Vi11}. 
Dividing it by a normalizing 
factor we quantify an optimal lower bound on the relative 
distance between the elements with respect to the normalizing factor. 
In \cite{CV13} the normalizing factor is
used to express this similarity on a scale of 0 to 1. For the multiset NGD
we use the normalizing factor to express
the similarity among the elements of $X$ on a scale from 0 to $|X|-1$,
Lemma~\ref{lem.0}. This is 
the shared similarity among all search terms in $X$
or relative semantics that all search terms in $X$ have in common.
The normalizing term $\max_{x \in X}\{G(x)\}$ yields \eqref{eq.NGD} 
coinciding with \eqref{eq.ngd} for the case $|X|=2$, as does setting
the denominator to $\max_{x \in X}\{G(X \setminus \{x\})\}$. However,
the choice
of $\max_{x \in X}\{G(x)\}$ works better for 
the applications (Section~\ref{sect.appl}).
\end{remark}
\begin{definition}
\rm
Define the {\em normalized Google distance} (NGD) by
\begin{eqnarray}\label{eq.NGD}
 NGD(X)& =&\frac{G(X) - \min_{x \in X}\{G(x)\}}
{\max_{x \in X}\{G(x)\}}
\\&=&  \frac{ \max_{x \in X} \{\log f(x)\}  - \log f(X)}{
\log N - \min_{x \in X}\{\log f(x)\}},
\nonumber
\end{eqnarray}
where $f(x) = |{\bf x}|$ denotes the number of pages (frequency) 
containing $x$, and $f(X)=|{\bf X}|$
denotes the number of pages (frequency) 
containing all elements (search terms)
of $X$, as reported by Google.
\end{definition}
The second equality in \eqref{eq.NGD}, 
expressing the NGD in terms of frequencies,
is seen as follows. We use \eqref{eq.gpmf} and \eqref{eq.gcc}. 
The numerator is rewritten by $G(X)=\log 1/g(X)=\log (N/f(X))=
\log N - \log f(X)$ and $\min_{x \in X}\{G(x)\}=
\min_{x \in X}\{ \log 1/g(x)\} 
= \log N - \max_{x \in X}\{\log f(x)\}$.   
The denominator is rewritten as $\max_{x \in X}\{G(x)\}
= \max_{x \in X}\{\log 1/g(x)\}
= \log N - \min_{x \in X}\{\log f(x)\}$.
For $X'=X \bigcup \{x\}$ the denominator of $NGD(X')$ is unchanged for 
$G(x) \leq \max_{y \in X}\{G(y)\}$, while it becomes larger otherwise.  
The last case means that $f(x) < \min_{y \in X} \{f(y)\}$, that is,
$x$ is more obscure than any $y \in X$. As $|X|$ rises this 
becomes more unlikely. 

\begin{remark}
\rm
In practice, use the page counts
returned by Google for the frequencies, and we have to
choose $N$.  From the right-hand side term in \eqref{eq.NGD}
it is apparent that by increasing $N$ we decrease the NGD , everything gets 
closer together, and
by decreasing $N$ we increase the NGD, everything gets further apart.
Our experiments suggest that every reasonable
($M$ or a value greater than any $f(x)$ in the formula) value can be used as
normalizing factor  $N$,
and our
results seem  in general insensitive to this choice.  In our software, this
parameter $N$ can be adjusted as appropriate, and we often use $M$ for $N$.
\end{remark}

\section{Theory}\label{sect.theory}
The following two lemmas state properties
of the NGD (as long as we choose parameter $N \geq M$). The proofs 
follow directly from the formula \eqref{eq.NGD}.
\begin{lemma}\label{lem.0}
Let $X$ be a multiset.
The {\em range} of the $NGD(X)$ is in between 0 and $|X|-1$ with the 
caveats in the proof. The upper bound states that $NGD(X)$ is at most
the number of elements of $X$ minus 1.
\end{lemma}
\begin{proof}
Clearly, $f(x) \geq f(X)$ for $x \in X$. Hence the numerator of
\eqref{eq.NGD} is nonnegative. By assumption above $N \geq M$. Hence
the NGD is nonnegative.
If $\max_{x \in X} \{\log f(x)\} = \log f(X)$ then $NGD(X)=0$ according
to \eqref{eq.NGD}. For example, let
$X=\{x,y\}$ and choose $x \neq y$
with ${\bf x}={\bf y}$. Then,
$f(x)=f(y)=f(x,y)$ and $NGD(x,y)=0$.
(In the practice of Google we have to deal with the effect mentioned in
beginning of Section~\ref{sect.gsd}. If for instance
$X=\{\mbox{red}, \mbox{red}, \mbox{red}\}$, then the numerator
of $NGD(X)$ equals $4.260.000.000- 5.510.000.000 = -1.250.000.000$.)

To obtain the upper bound
of $|X|-1$ we have 
$\log N - \min_{x \in X} \{\log f(x)\} \geq
(|X|-1)(\max_{x \in X} \{\log f(x)\}- \log f(X))$.
This implies $f(X) \geq 1$ since otherwise the right-hand side is $\infty$. 
Equivalently, $G(X)-\min_{x \in X}\{G(x)\} \leq (|X|-1)\max_{x \in X}\{G(x)\}$
and $G(X) < \infty$. This holds since if $X=\{x_1, \ldots, x_n\}$
then $G(X) \leq G(x_1) \ldots G(x_n)$. 
Equality is reached for $X$ consists of elements such that 
compressing them together results in no reduction of length over
compressing them apart while the compressed version of every element 
has the length of the maximum.
\end{proof}

We achieve $f(X)=0$  in case $\{x,y\} \subseteq X$ and $x$ and $y$
never occur on the same web pages
while both do occur on different pages. Then the numerator of $NGD(X)=\infty$. 
(In \cite{CV07} the upper bound of the range of the pairwise NGD ($|X|=2$) is
stated to be $\infty$. However, it is 1 with above restrictions.) 

\begin{lemma}
If for some $x$ we have that the frequency  $f(x)=0$,
then for every set $X$ containing $x$ as an element we have $f(X)=0$,
and the $NGD(X)= \infty / \infty$,
which we take to be $|X|-1$.
\end{lemma}

Metricity of distances $NID(X)$
and $NCD(X)$ of multisets $X$
are defined and treated in \cite{Vi11,CV13}. 
For every multiset $X$ we have that the $NGD(X)$ is invariant under 
permutation of $X$: it is symmetric. 
But the NGD distance is not a metric. 
Metricity is already violated for multisets of
cardinality two since $NGD(x,y)$ can be 0 with $x \neq y$. 
Hence there is no need to repeat the general metric axioms 
for nonempty finite multisets of more then two elements.
\begin{theorem}\label{theo.triangle}
The NGD distance is not a metric.
\end{theorem}
\begin{proof}
By Lemma~\ref{lem.0}  we have
$NGD(X) \not> 0$ for every $X$ consisting of unequal elements. 
Nor does the NGD satisfy the triangle inequality. Namely,
$NGD(x,y) \not\leq 
NGD(x,z)+NGD(z,y)$ 
for some $x,y,z$. For example,
choose ${\bf z}= {\bf x} \bigcup {\bf y}$,
$|{\bf x} \bigcap {\bf y}|= 1$,
${\bf x}={\bf x} \bigcap {\bf z}$,
${\bf y}={\bf y} \bigcap {\bf z}$,
and $|{\bf x}|=|{\bf y}|= \sqrt{N}$.
Then, $f(x)=f(y)=f(x,z)=f(y,z)=\sqrt{N}$, $f(z)=2\sqrt{N} -1$, 
and $f(x,y)=1$.
This yields $NGD(x,y)= 1$ and
$NGD(x,z)=
NGD(z,y) \approx 2/ \log N$, 
which violates the triangle
inequality for all $N > 16$. 
\end{proof}

\begin{remark}
\rm
In practice for the combination of the World Wide Web and Google
the $NGD$ may satisfy the triangle inequality. We did not
find a counterexample.
\end{remark}

For definitions and results about Kolmogorov complexity and the 
universal distribution (universal probability mass function)
see the text \cite{LV08} and for the Kolmogorov complexity of multisets see
Section~\ref{sect.set}. 

The information distance between strings $x$ and $y$ is 
$\max\{K(x|y),K(y|x)\}$, \cite{BGLVZ98}. The {\em symmetry of information}
states that $K(x,y)=K(x)+K(y|x)$ with a logarithmic additive term \cite{Ga74}.
Ignoring this logarithmic additive term, 
the information distance equals $K(x,y)-\min{K(x),K(y)}$.
Similarly, the information distance (diameter)
of a multiset $X$ is $K(X)-\min_{x \in X}\{K(x)\}$.
In \cite{Vi11} it is proven that this generalization retains 
the same properties.
Normalizing this information distance similar to \eqref{eq.NGD}
yields
\begin{equation}\label{eq.nid}
\frac{K(X)-\min_{x \in X}\{K(x)\}}
{\max_{x \in X} \{K(x)\}}.
\end{equation}

\begin{theorem}\label{theo.ideal}
Let $X$ be a nonempty finite multiset of search terms.
The closer $g(X)$ approximates ${\bf m}({\bf X})$ 
the better \eqref{eq.NGD} approximates \eqref{eq.nid}
with equality for $g(X) = {\bf m}({\bf X})$.
In this sense $NGD(X)$ approximates
the shared similarity among all search terms in $X$
or relative semantics that all search terms in $X$ have in common. 
\end{theorem}
\begin{proof}
The theoretic foundation of \eqref{eq.NGD} lies in the fact that
\eqref{eq.NGD} approximates \eqref{eq.nid} because $G$ is viewed 
as a compressor.
The quantity
$G(X)$ semicomputes the Kolmogorov complexity $K({\bf X})$
from above. 
This means by Definition~\ref{def.g} that 
the Google probability $g(X)$
semicomputes the universal probability ${\bf m}({\bf X})$ 
(the greatest effective 
probability) from below. Namely, by the Coding Lemma 
of L.A. Levin \cite{Le74} (Theorem 4.3.3 in \cite{LV08}),
the negative logarithm of the universal probability ${\bf m}$
equals (up to a constant additive term) the prefix 
Kolmogorov complexity $K$ of the same finite object:
\[
\log 1/{\bf m}({\bf X}) = K({\bf X})+O(1).
\]
From Shannon \cite{Sh48} we know that
the negative logarithm of every probability mass function equals
the code length of a corresponding code up to a constant additive term. 
By the last displayed equation
it turns out that the negative logarithm of the greatest lower 
semicomputable probability mass function equals up to
a constant additive term the least length of an upper semicomputable
prefix code of the same finite object
(the prefix Kolmogorov complexity \cite{LV08}).
 
This means that the closer the Google probability approximates the
universal probability of the corresponding event, 
the better the NGD approximates its theoretical
precursor \eqref{eq.nid}.
The latter is the same  formula as the former
but with the prefix Kolmogorov complexity $K$ instead of the 
Google code $G$. Namely, 
\begin{equation}\label{eq.NID}
 \frac{K({\bf X})-\min_{x \in X}\{K({\bf x})\}}
{\max_{x \in X} \{K({\bf x)\}}},
\end{equation}
is the same as \eqref{eq.nid} except that 
${\bf x}$ means the event corresponding with
the search term $x$. 
The Kolmogorov complexity
$K({\bf X})$ is a lower bound on $Z({\bf X})$ 
for every computable compressor $Z$,
known and unknown alike. Hence also for $G$. 
\begin{remark}\label{rem.nid-ngd}
\rm
Let $X$ be a multiset of finite binary objects rather
than a multiset of the names of finite binary objects.
For $|X| \geq 2$ the formula 
\eqref{eq.nid}, approximates the similarity
all the elements in $X$ have in common.
This is a consequence of the fact that in
Theorem 5.2 of \cite{Vi11} it is proven that 
$E_{\max}(X)=\max_{x\in X} K(X|x)$ (up to a $O(\log K(X))$ additive term)
is universal among the class of 
admissible multiset distance functions.
Here `universal' means it is (i) a member of the class, and (ii)
its distance is at most an additive constant less 
than that of any member in the class. 
Multisets with admissible distances are a very 
wide class of multisets; so wide that 
presumably all distances of multisets one cares about are included. 
\end{remark}

In the Google
setting we replace the argument $X$ in the above 
Remark~\ref{rem.nid-ngd} by the 
event ${\bf X}$.  Formally,
one defines the class of admissible distances in this setting as follows.
A priori we allow asymmetric distances. We
would like to exclude degenerate distance measures 
such as $D({\bf X})=1/2$
for all $X$ containing a fixed element $x$. In general,
for every $d$, we want only finitely many
sets $X \ni x$ such that $D({\bf X}) \leq d$. Exactly how fast we want the
number of sets we admit to go to $\infty$ is not important; it is only a
matter of scaling. 
\begin{definition}
\rm
A distance function $D$ on nonempty finite multisets ${\bf X}$ 
is {\em admissible} if
(i) it is a total, possibly asymmetric, function from ${\bf X}$ 
to the nonnegative real numbers which is 0 if there is an $Y$
with $|Y|=1$ and ${\bf Y}={\bf X}$, and
greater than 0 if all $X$'s corresponding with ${\bf X}$
have $|X| > 2$
(up to a $O(\log K)$ additive term 
with $K=K({\bf X})$), (ii) is upper
semicomputable (see below), and (iii) satisfies
the density requirement for every $x \in \{0,1\}^*$ given by
\[
\sum_{X: x \in X \& D({\bf X})>0} 2^{-D({\bf X})} \leq 1.
\] 
\end{definition}
\begin{remark}
\rm
Here, we consider only distances that are computable in some broad
sense (an upper semicomputable object means
that the object can be computably approximated
from above in the sense that in the limit the object is reached).
In the Google case the distance of $X$ is computable from ${\bf X}$.
See for example \cite{LV08,Vi11}.
\end{remark}
With adapted notation (because we go from compressible strings to 
queries for Google and the corresponding events) 
Theorem 5.2 of \cite{Vi11} states: $E_{\max}({\bf X})$ is admissible
and it is minimal in the sense that for every admissible
multiset distance function $D( {\bf X})$ we have 
$E_{\max}( {\bf X}) \leq D( {\bf X})$ up to
an additive constant term.
\begin{remark}
\rm
For example, assume
that $X=\{x_1, \ldots ,x_n\}$
and we take
a quantifiable feature and its 
quantification over ${\bf x}_1, \ldots , {\bf x}_n$ 
is $f_1, \ldots , f_n$. If the feature is absent in ${\bf x}_i$ then $f_i=0$
($1 \leq i \leq n$). Such a feature can also exist of
an upper semicomputable arithmetic combination of features.
Consider $m_f = \min_i \{f_i\}$. Let
an admissible distance $D({\bf X})$ be defined in terms of 
$M_f =\max \{m_f: \:f$ is a feature in 
all $x \in X \}$. If $D({\bf X})$ is equal to 0, then 
there is a $Y=\{y_1,\ldots y_m\}$ with ${\bf Y}={\bf X}$ such that
all elements ${\bf y}_i$ of ${\bf Y}$ are equal.
In this case $E_{\max}({\bf X})=0$ as well.
In fact, since $E_{\max}( {\bf X})$ is an admissible distance it
is the least such distance and incorporates admissible distances 
of all features.
\end{remark}
The numerator of \eqref{eq.NID} is $E_{\max}({\bf X})$ 
(up to a logarithmic additive term
which is negligible), and the denominator is
a normalizing factor
${\max_{x \in X} \{K({\bf x})}$.
Therefore \eqref{eq.NID} expresses a greatest lower bound
on the similarity on a scale from 0 to $|X|-1$ among
the elements ${\bf x}_i$ of ${\bf X}$ according to the 
normalized admissible distances.
Hence, the closer $g(X)$ approximates ${\bf m}({\bf X})$
from below the better $NGD(X)$ approximates \eqref{eq.NID}
with equality for $g(X) = {\bf m}({\bf X})$.
\end{proof}

\section{Applications}\label{sect.appl}
The application of the approach presented here requires the ability to query a database for the number of occurrences and co-occurrences of the elements in the multiset that we wish to analyze. One challenge is to find a database that has sufficient breadth as to contain a meaningful numbers of co-occurrences for related terms. As discussed previously, an example of one such database is the World Wide Web, with the page counts returned by Google search queries used as an estimate of co-occurrence frequency. There are two issues with using Google search page counts. The first issue is that Google limits the number of programmatic searches   in a single day to a maximum of 100 queries, and charges for queries in excess of 100 at a rate of up to \$50 per thousand. This limit applies to searches processed with the Google custom search application programmers interface (API). To address this constraint, we present results for multisets of cardinality small enough to allow a single analysis to be completed with less than 100 Google queries. For a multiset of cardinality $n$, the number of queries required for a pairwise NGD to form a full pairwise distance matrix using the symmetry of the NGD as discussed in Section \ref{sect.theory} requires us to to compute the entire upper triangular portion of the distance matrix to asses the distances among all pairs out of $n$ elements, plus the $n$ diagonal singleton elements. This requires $(n^2-n)/2+n$ queries. For the NGD of multisets, the number of queries required for an set of cardinality $n$ is $2*n*c+n$,  where $c$ is the number of classes. This number of queries results from 2 calculations for each element $x$  and class $A$, $NGD(Ax)$ and $NGD(A)$, and $n$ calculations for the singleton elements as described in \cite{CV13}, eqn. VIII.1. (We write $Ax$ for $A \bigcup \{x\}$.)  In a leave-one-out cross validation, results from these queries can be precomputed for each out of $c$ classes and for each individual search element out of $n$ of those, reducing the number of queries to $n*c+c+n$.  The multiset formulation has an advantage over the pairwise of requiring less queries. 
For example, for $c=2$ and $n>3$ the pairwise case requires more search queries 
than the multiset case.

The second issue with using Google web search page counts is that the numbers are not exact, but are generated using an approximate algorithm that Google has not disclosed. For the questions considered in the current work and previously \cite{CV07} we have found that these approximate measures are sufficient to generate useful answers, especially in the absence of any a priori domain knowledge. 

It is possible to implement internet based searches without using search engine API's, and therefore not subject to daily limits. Any internet search that returns a results count can be used in computing the NGD. This includes the Google web page interface as well as the Bing search engine. The results from any other custom internet searches, such as the NIH PubMed Central site can also be parsed to extract the page count results. For the sample applications discussed here we present results using the Google search API, as well as results using the Google and Bing web user interface. Interestingly, the Google web user interface returns different results compared to the Google API, with the Google API performing better in some cases and equivalent in others.

One alternative to using search queries on the World Wide Web is to use a large database such as the Google books n-gram corpus \cite{Michel11}. The n-gram corpus computes word co-occurrences from over 8 million books, or nearly 6\% of all the books ever published. Co-occurrences achieving a frequency of less than 40 are omitted. This data is available for download free of charge, although the sheer size of the data set does make it challenging to download and parse in a time efficient manner. The n-grams data set is available for [1-5]-grams. For the studies described here, the 5-gram files were used for co-occurrence counts and the 1-gram files were used for occurrence counts. The files were downloaded and decompressed as needed. They were then preprocessed using the Unix 'grep' utility to preserve only the relevant terms. Finally, a MATLAB script for regular expression analysis was used to extract co-occurrence counts. Despite the large size of the database, there was not enough  co-occurrence data to extract useful  counts for the  calculating the  multiset distance using the n-gram data. Results were  obtained only for the pairwise analysis.
While in the regular use of the NGD \eqref{eq.ngd} and \eqref{eq.NGD}
the frequency counts of web pages are based on 
the co-occurrence of search terms on that web page, in the n-gram
case the further condition holds that the co-occurrences must be
within $n$ words. Thus, one substitutes the frequency counts in
the NGD formulas with a frequency count that is less or equal.

Here we describe a comparison of the NGD using the multiset formulation based on Google web search page counts with the pairwise NGD formulation based on Google API and web search page counts, with Bing web search interface page counts and also with a pairwise formulation of the NGD based on Google books n-gram corpus. We consider example classification questions that involve partitioning a set of words into underlying categories. For the NGD multiset classification, we proceed as in \cite{CV13}, determining whether to assign element $x$ to class $A$ or class $B$ by computing $NCD(Ax)-NCD(A)$ and $NCD(Bx)-NCD(B)$ and assigning element $x$ to whichever class achieves the minimum. Although the NGD works with multisets of cardinality two or larger, the classification approach used here requires that each class have cardinality of three or larger so that during cross validation there is never a scenario where we need to compute the NGD for a set of cardinality one.

For the pairwise formulation, we use the gap spectral clustering unsupervised approach developed in \cite{CO09}. 
Gap spectral clustering uses the gap statistic as first proposed in \cite{Tib01} to estimate the number of clusters in a data set from an arbitrary clustering algorithm. In \cite{CO09}, it was shown that the gap statistic in conjunction with a spectral clustering \cite{Ng02} of the distance matrix obtained from pairwise NCD measurements is an estimate of randomness deficiency for clustering models. Randomness deficiency is a measure of  the meaningful information that a model, here a clustering algorithm, captures in a particular set of data \cite{LV08}. The approach is to select the number of clusters that minimizes the randomness deficiency as approximated by the gap value. In practice, this is achieved by picking the first number of clusters where the gap value achieves a maximum as described in \cite{CO09}. 

The gap value is computed by comparing the intra-cluster dispersions of the pairwise NGD distance matrix to that of uniformly distributed randomly generated data on the same range. For each value of $k$, the number of clusters in the data, we apply a spectral clustering algorithm to partition the data, assigning each element in the data to one of $k$ clusters. Next, we compute $D_r$, the sum of the distances between elements in each cluster $C_r$,
\[{{D}_{r}}=\sum\limits_{i,j\in {{C}_{r}}}{{{d}_{i,j}}}.\]
The average intra-cluster dispersion is calculated,
\[{{W}_{k}}=\sum\limits_{r=1}^{k}{\frac{1}{2{{n}_{r}}}}{{D}_{r}},\]
where $n_r$ is the number of points in cluster $C_r$.
The gap statistic is then computed as the difference between the intra-cluster distances of our data and the intra-cluster distances of $B$ randomly generated uniformly distributed data sets of the same dimension as our data,
\[Gap(k)=\frac{1}{B}\sum\limits_{b=1}^{B}{\log ({{W}_{kb}}})-\log ({{W}_{k}}),\]
where ${{W}_{kb}}$ is the average intra-cluster dispersion obtained by running our clustering algorithm on each of the $B$ randomly generated uniformly-distributed datasets. Following \cite{CO09} we set $B$ to 100. We compute the standard deviation of the gap value $s_k$ from $\sigma_k$, the standard deviation of the $B$ uniformly distributed randomly generated data, adjusted to account for simulation error, as
\[{{s}_{k}}={{\sigma }_{k}}\sqrt{1+{}^{1}/{}_{B}}.\]
Finally, we choose the smallest value of $k$ for which
	\[Gap(k)\ge Gap(k+1)-{{s}_{k+1}}.\]

We now describe results from a number of  sample applications. For all of these applications, we use a single implementation based on co-occurrence counts. We compare the results of the NGD multiset, the NGD pairwise used with gap spectral clustering, and for the first question, the Google n-grams data set used pairwise with gap spectral clustering. The n-grams data set, despite being an extremely large database, did not have the scope to be used with the multiset formulation, or with the more detailed sample applications. We include results using the Google API and the Google and Bing web user interfaces. The results are summarized in Figure \ref{tab.AllResults}.

\begin{figure}[htbp]
\begin{center}
\includegraphics{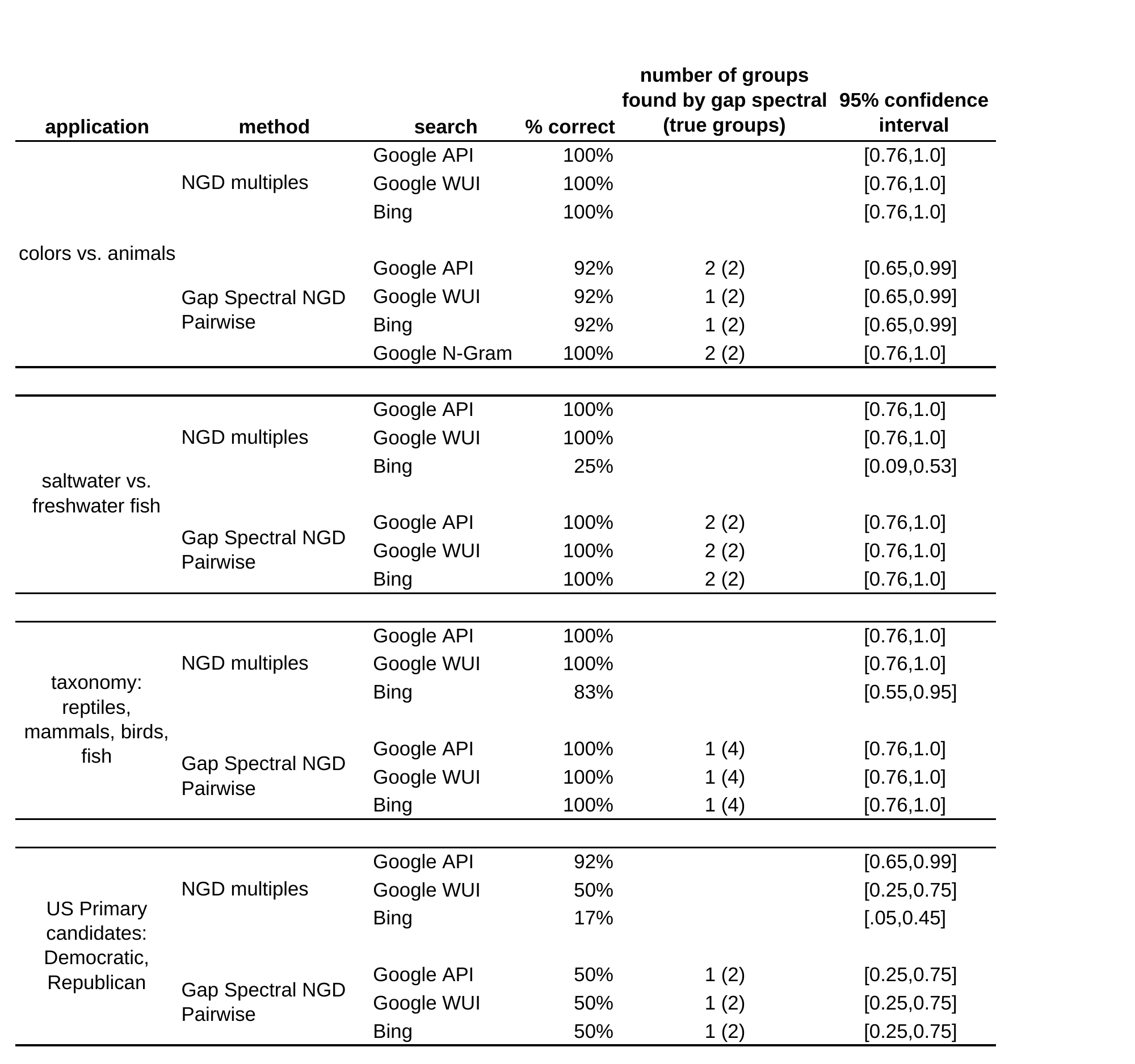}
\end{center}
\caption{Summary of results for different applications, methods and search interfaces.}
\label{tab.AllResults}
\end{figure}

In the first question we consider classifying colors vs. animals : \{'red', 'orange', 'yellow', 'green', 'blue', 'indigo'\} vs. \{'lion', 'tiger', 'bear', 'monkey', 'zebra', 'elephant'\}. For this question, with the pairwise NGD, gap spectral clustering found two groups in the data and classified all of the elements correctly  except  that 'indigo' was classified  with the  animal class. The NGD multiset formulation classified the terms perfectly. Using the pairwise NGD with the n-gram corpus gap spectral clustering found two groups in the data and also classified the terms perfectly. Figure \ref{tab.results} shows the distance matrices computed for the three approaches. We obtained identical results using the Google API as with the Bing and Google web user interfaces, although for the Gap Spectral clustering, the Bing and Google web user interface results did not find the correct number of groups in the data.

\begin{figure}[htbp]
\begin{center}
\includegraphics{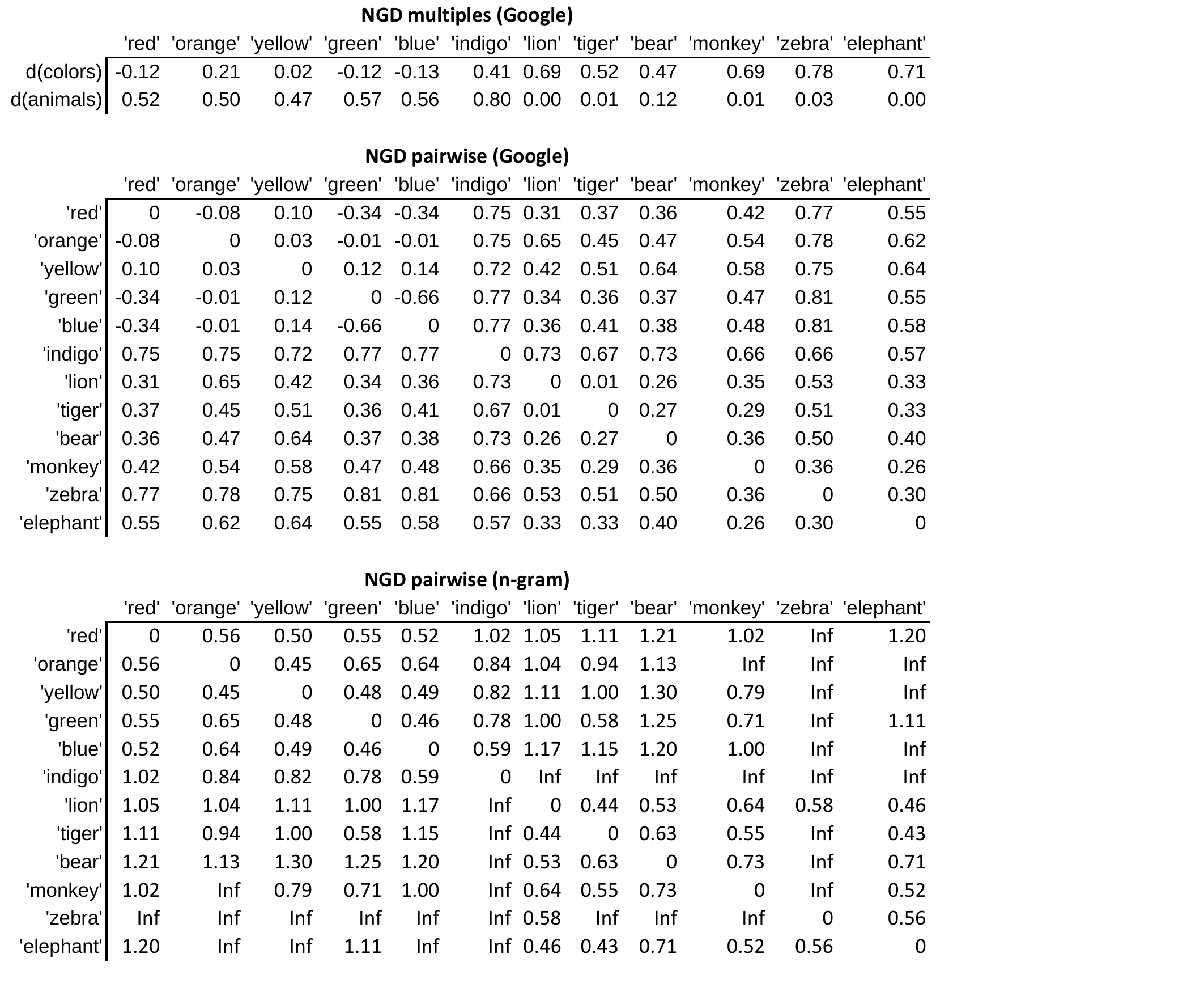}
\end{center}
\caption{Distances computed for the animals vs. colors question. The distances computed via leave-one-out cross validation with the NGD for multisets using Google API page counts (top). Distances computed using the pairwise NGD with Google API page counts (middle). Distances computed using pairwise Google Books n-grams co-occurrences (bottom).}
\label{tab.results}
\end{figure}

The next question considered compared saltwater vs. freshwater fish: \{'Labidochromis caeruleus', 'Sciaenochromis fryeri', 'Betta splendens', 'Carassius auratus', 'Melanochromis cyaneorhabdos'\} vs. \{'Ecsenius bicolor', 'Pictichromis paccagnellae', 'Amphiprion ocellaris', 'Paracanthurus hepatus', 'Chromis viridis'\}. For this question, using the pairwise NGD, gap spectral clustering found two groups in the data and classified all of the elements correctly.   The NGD multisets formulation also classified the terms perfectly. The pairwise NGD with the n-gram corpus did not have enough data to classify any of the elements for this question or any of the remaining questions. The Bing web user interface performed poorly in the NGD multisets formulation, classifying only 3 of 12 elements correctly.

In a related question, we next consider a 4 class problem based on subclasses of reptiles, mammals, birds and fish. We consider 12 taxonomic subclasses: \{'Ichthyosauria', 'Lepidosauromorpha', 'Archosauromorpha'\} (reptiles), \{'Allotheria', 'Triconodonta', 'Holotheria'\} (mammals), \{'Agnatha', 'Chondrichthyes', 'Placodermi'\} (fish) and \{'Archaeornithes', 'Enantiornithes', 'Hesperornithes'\} (birds). The NGD for multisets classified all 12 subclasses correctly. Again using the Bing web user interface with the multisets distance the classification accuracy was significantly worse than with Google search. Using the NGD pairwise, the gap statistic found only one group in the data. This was a local maxima, there was a strong global maximum at 4 groups in the data as shown in Figure \ref{fig.gap}. Clustering into the four groups with spectral clustering classified all elements correctly.

Finally, we considered the 12 candidates from the republican and democratic primaries in the 2008 US presidential election, \{'Barack Obama', 'Hillary Clinton','John Edwards', 'Joe Biden', 'Chris Dodd','Mike Gravel'\} (Democratic) and \'John McCain', 'Mitt Romney', 'Mike Huckabee','Ron Paul', 'Fred Thompson', 'Alan Keyes'\} (Republican). The candidates were ordered by the length of time they remained in the race. The NGD multisets classified 11 of the 12 correctly, misclassifying 'Ron Paul' with the Democrats, possibly a result of his strongly libertarian political views. Gap spectral clustering with the pairwise NGD found only one group in the data. Spectral clustering the candidates into two groups was not effective at classification, grouping the candidates more by popularity than by party, \{'Barack Obama', 'Hillary Clinton', 'Joe Biden', 'John McCain', 'Mitt Romney', 'Ron Paul'\} and \{'John Edwards', 'Chris Dodd', 'Mike Gravel', 'Mike Huckabee', 'Fred Thompson','Alan Keyes'\}. The Google API was more accurate than the Google web interface for this question, and again the Bing search engine performed poorly relative to Google in the multisets formulation.

\begin{figure}[htbp]
\begin{centering}
\includegraphics{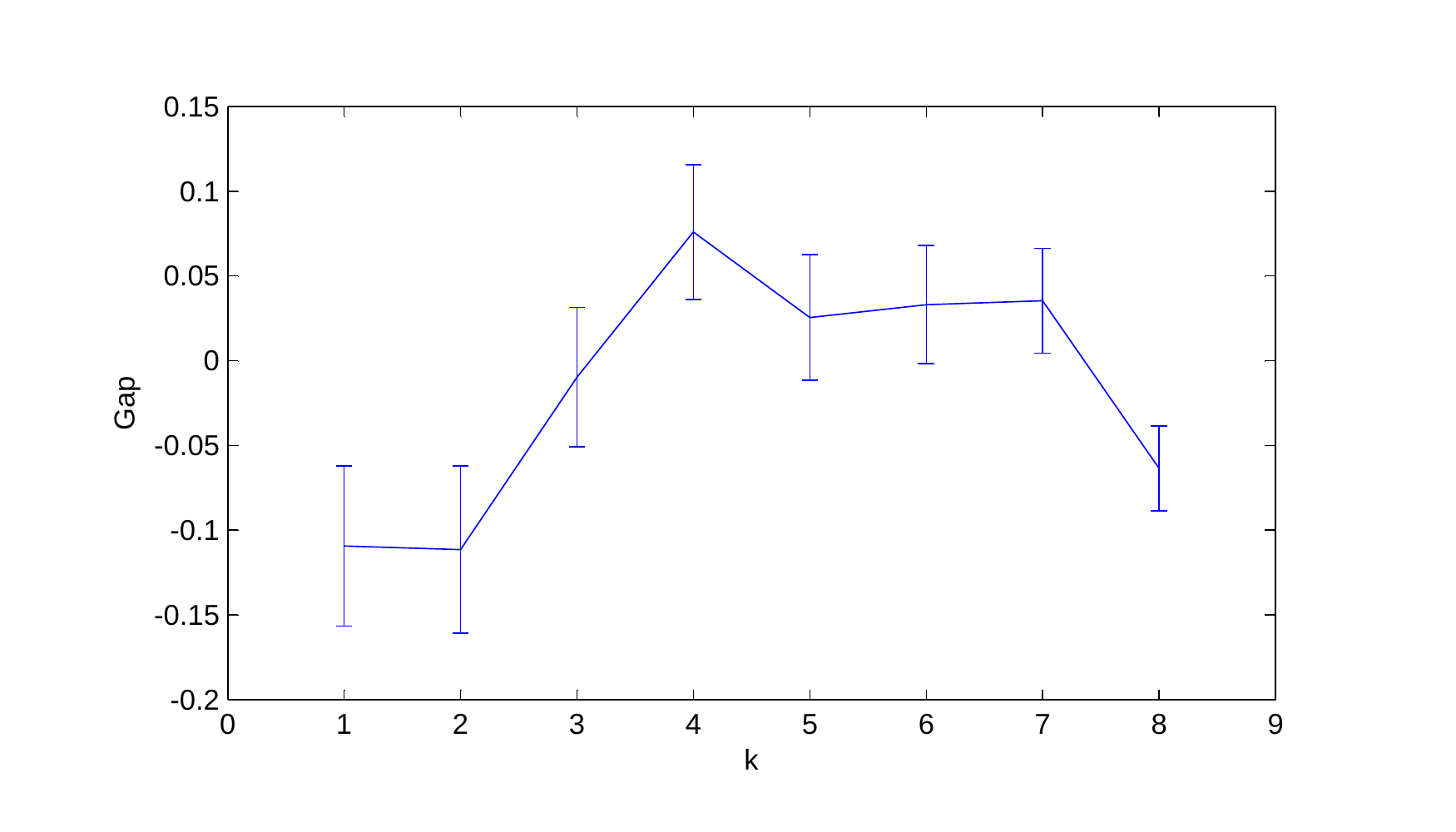}
\end{centering}
\caption{Plot of gap statistic for the four class taxonomy question. The gap selected the local maxima at 1. The global maxima seen at 4 is the correct answer.}
\label{fig.gap}
\end{figure}

\section{Conclusions and Discussion}\label{sect.concl}

We propose a method, the normalized Google distance (NGD) for multisets,
say $X$, that quantifies in a single number between 0 and $|X|-1$ 
the way in which
a finite nonempty multiset of objects is similar. This similarity
is based on the names of the objects and uses the information
supplied by the frequency counts of pages on which the names
occur. For instance, the pages are pages of the World Wide Web
and the frequency counts are supplied by a search 
engine such as Google. The method can be applied using any big data
base (Wikipedea, Oxford English Dictionary) and a search engine
that returns aggregate page counts. Since this method uses names for 
object, and not the objects themselves, we can view those names
as words and the similarity between words in a multiset 
of words as a shared relative
semantics between those words. (What is stated about ``words'' 
also holds for ``phrases.'')
The similarity between a finite nonempty multiset of
words is called a distance or diameter of this multiset. 
Using the theory developed we show properties
of this distance and especially that it is not a metric
(Theorem~\ref{theo.triangle}), in certain cases
the distance may be 0 for distinct names, 
in non-pathological cases it ranges in between 0 and $|X|-1$ (for
a multiset $X$), and does not
satisfy the triangle property. However, in the 
(World Wide Web, Google) practice we did not find examples
of violating the triangle property.  We showed that the closer the probability distribution  
of names supplied by the (data base, search engine) pair
is to the universal distribution, the closer the computed similarity
is to the ultimate effective similarity (the ``real'' similarity) 
with equality in the limit
(Theorem~\ref{theo.ideal}). For instance, the World Wide Web is 
so large and Google so good that the two are close. 

Earlier \cite{CV07} this approach was used for pairs of names,
the pairwise NGD,
an subsequently widely used, see the discussion in 
Section~\ref{sect.intro}. To test
the effacity of the new method we compared its 
results (Section~\ref{sect.appl}) 
with those of the pairwise method (together with some embellishments)
on small data sets, to wit colors versus animals, 
saltwater versus freshwater fish, taxonomy of reptiles mammals birds
fish, and US Primary candidates: Democratic versus Republican.
We used the World Wide Web and Google API, Google WUI, Bing,
and Google n-grams. The results showed in these instances
superiority or equality for the multiset NGD over the
pairwise NCD for Google. Bing performed poorly and the n-gram method
was only usable for the pairwise NGD in view of the fact
that $n \leq 5$ only and the cardinality of the test multisets were
too large.

\bibliographystyle{plain}

\end{document}